\documentclass[submission,copyright,creativecommons]{eptcs}
\providecommand{\event}{GandALF 2020}

\usepackage{times}
\usepackage{breakurl}             

\usepackage{graphicx}
\graphicspath{{/home/awells/Documents/probabilistic_verification/figs}}
\usepackage{xcolor}

\usepackage[algoruled,vlined,linesnumbered]{algorithm2e}

\usepackage{amsmath}
\usepackage{amssymb}

\usepackage{tikz}
\usetikzlibrary{automata,positioning,angles,quotes}

\usepackage{microtype}

\usepackage{caption}
\usepackage{subcaption}

\usepackage{wrapfig}

\usepackage{amsthm}
\newtheorem{theorem}{Theorem}
\newtheorem{definition}{Definition}
\newtheorem{lemma}{Lemma}
\newtheorem{corollary}{Corollary}

\newtheorem{problem}{Problem}

\usepackage[title]{appendix}

\newcommand{\ltl}{\textsc{LTL}\xspace}
\newcommand{\ltlf}{$\textsc{LTL}_f$\xspace}
\newcommand{\ltlphi}{\varphi}
\newcommand{\ltlfphi}{\phi}
\newcommand{\M}{\mathcal{M}}
\newcommand{\term}{\text{term}}
\newcommand{\init}{\text{init}}
\newcommand{\A}{\mathcal{A}}
\renewcommand{\L}{L}
\newcommand{\Act}{A}
\newcommand{\mypath}{w}
\newcommand{\trace}{\rho}

\usepackage[algoruled,vlined,linesnumbered]{algorithm2e}

\usepackage[final]{changes}

\newcommand{\shrinkb}{\def\baselinestretch{0.995}\large\normalsize}
\begin{document}
\shrinkb

\title{\ltlf Synthesis on Probabilistic Systems}



\author{Andrew M. Wells
\institute{Rice University, Houston, Texas}
\email{andrew.wells@rice.edu}
\and
Morteza Lahijanian
\institute{University of Colorado, Boulder, Colorado}
\email{morteza.lahijanian@colorado.edu}
\and
Lydia E. Kavraki
\institute{Rice University, Houston, Texas}
\email{kavraki@rice.edu}
\and
Moshe Y. Vardi
\institute{Rice University, Houston, Texas}
\email{vardi@rice.edu}
}
\def\titlerunning{\ltlf Synthesis on Probabilistic Systems}
\def\authorrunning{Wells et al.}



\maketitle


\begin{abstract}
 Many systems are naturally modeled as Markov Decision Processes (MDPs), combining probabilities and strategic actions. Given a model of a system as an MDP and some logical specification of system behavior, the goal of synthesis is to find a policy that maximizes the probability of achieving this behavior.   A popular choice for defining behaviors is Linear Temporal Logic (\ltl). Policy synthesis on MDPs for properties specified in \ltl  has been well studied. \ltl, however, is defined over infinite traces, while many properties of interest are inherently finite. Linear Temporal Logic over finite traces (\ltlf) has been used to express such properties, but no tools exist to solve policy synthesis for MDP behaviors given finite-trace properties.  We present two algorithms for solving this synthesis problem: the first via reduction of \ltlf to \ltl and the second using native tools for \ltlf. We compare the scalability of these two approaches for synthesis and show that the native approach offers better scalability compared to existing automaton generation tools for \ltl.
\end{abstract}

\section{Introduction}
\label{sec:introduction}



Many real-world systems are stochastic in nature.  They evolve in the world according to action (control) decisions and the uncertainty embedded in the execution of those actions, resulting in stochastic behavior. \textit{Formal synthesis} studies how the system should choose actions so that it can increase the chances of achieving a desirable behavior.  To allow such reasoning, \textit{Markov Decision Processes} (MDPs) are typically used to model these systems since MDPs effectively capture sequential decision-making and probabilistic evolutions \cite{baier2008principles}.  The desired behavior is expressed in a formal language, which yields expressive and unambiguous specifications.  Temporal logics are a common choice for this language 
since they allow a combination of temporal and boolean reasoning over the behavior of the system.  Most temporal logics are interpreted over behaviors with infinite time durations, but some behaviors are inherently finite and can be expressed more intuitively and more practically using a temporal logic with finite semantics \cite{He:IROS:2017,he2019efficient}.  This work investigates policy synthesis on MDPs for specifications expressed in a formal language called \textit{Linear Temporal Logic over finite traces} (\ltlf) \cite{de2013linear}, which reasons over system behaviors with finite horizons.


A popular specification language in formal verification is Linear Temporal Logic (\ltl) \cite{pnueli1977temporal}.
\ltl provides a natural description of temporal properties over an infinite trace.  \ltl can express properties such as order, e.g., ``first $a$ and next eventually $b$,'' and lack of starvation, e.g., ``globally eventually resources are granted.''  While infinite-trace properties are essential for reasoning about many systems, other systems require finite-trace properties \cite{He:IROS:2017,he2019efficient}, for example robot planning for finite behaviors. If a robot is to build an object, the interest is in the finite rather than the infinite traces that accomplish this task.  In such cases \ltlf \cite{de2013linear} is a suitable alternative to \ltl, cf.~\cite{he2019efficient}.

A number of recent studies has focused on developing efficient frameworks for solving \ltlf-reasoning problems, e.g., \cite{de2013linear,de2015synthesis,zhu2017symbolic,He:IROS:2017,Li:AAAI:2019,he2019efficient}.  De Giacomo and Vardi \cite{de2013linear} presented a translation of \ltlf to \ltl, implying that existing tools for \ltl satisfiability and synthesis can be used for \ltlf with suitable transformations.  In \cite{de2015synthesis}, the problem of reactive synthesis from \ltlf specifications that contain system (controllable) and environment (uncontrollable) variables is considered.  
The work shows that this problem reduces to strategy synthesis in a two-player automaton game, which is a 2EXPTIME-complete problem.  To enable practical solutions to this problem, \cite{zhu2017symbolic} introduces a symbolic approach to the \ltlf synthesis problem.  The applications of those results in the field of robot planning are studied in \cite{He:IROS:2017,He:RAL:2018,he2019efficient}. In \cite{Li:AAAI:2019}, the fundamental problem of \ltlf satisfiability checking is studied through a SAT-based approach. The results of those studies show that more scalable tools can be built by treating \ltlf \emph{natively}, rather than reducing to \ltl, even though \ltl reasoning is a more mature and extensively studied domain.  The underlying assumption in all those works is that the system is either purely deterministic or purely nondeterministic.  Hence, it is natural to ask whether similar results can be obtained for probabilistic systems. That is, do native \ltlf techniques outperform reasoning by reduction to \ltl?

Formal synthesis on probabilistic systems has been extensively studied in the formal-verification literature.  In those studies, MDPs are a de facto modeling tool for the probabilistic system, and \ltl and PCTL (probabilistic computation temporal logic) \cite{baier2008principles} are the specification languages of choice.  For \ltl synthesis, the approach is based on first translating the specification to a \textit{deterministic Rabin automaton} (DRA) and then solving a stochastic shortest path problem \cite{Var85b,de1997formal} on the product of the DRA with the MDP.  Tools such as PRISM \cite{kwiatowska2011prism} can solve this \ltl synthesis problem efficiently with their symbolic engine.  \ltlf synthesis for probabilistic systems, however, has not been studied, and it is not clear whether the same tools can be extended for \ltlf reasoning with a similar efficiency.

Assigning rewards based on temporal goals has also been studied in the context of planning. In \cite{thiebaux2006decision}, a logic similar to \ltlf, \$FLTL, is considered; however, this logic cannot express properties such as ``Eventually.'' In \cite{bacchus1996rewarding,bacchus1997structured}, PastLTL is used to describe finite properties associated with rewards. PastLTL and \ltlf naturally express different properties (see \cite{thiebaux2006decision}). In \cite{camacho2017nonmarkovian,brafman2018ltlf}, rewards can be attached to \ltlf or \textit{linear dynamic logic on finite traces} (LDL$_f$) \cite{de2013linear} formulas. In all of these MDP planning works, the goal is to maximize rewards rather than to study the behavior of the MDP. These approaches use approximations that give 
lower bounds and \replaced{converge}{converging} only in the limit.
Thus, they cannot necessarily give a negative answer to a decision query about a synthesis problem (e.g., does a policy with at least 95\% probability into success exist?). 

In this work, we present the problem of \ltlf synthesis for 
MDPs. 
In approaching this problem, we specifically seek to answer the empirical question of which approach is more efficient; an approach based on the mature and well-studied \ltl synthesis or an approach based \added{on} the unique properties of \ltlf itself. The answer to this question can have a broader impact on the employment of formal methods for probabilistic systems. Hence, we introduce two solutions to this problem.  The first one is based on a reduction to an \ltl synthesis problem, and the second one is a native approach. For the first approach, we show a translation of the \ltlf specification to \ltl and a corresponding augmentation of the MDP that allows us to use standard tools for \ltl synthesis. For the second approach, we use specialized tools to obtain an automaton that we input to standard tools in order to solve the problem. 

Even though both approaches have the same theoretical complexity bound \cite{baier2008principles,de2015synthesis}, 
we demonstrate that the native approach scales better than the translation to \ltl through a series of benchmarking case studies. 
Our code and examples are available on GitHub \cite{wells2020github}.

\section{Problem Formulation}
\label{sec:problem_formulation}

Our goal is to synthesize a policy for an MDP such that the probability of satisfying a given \ltlf property is maximized.  First we introduce the formalisms needed to define this problem. \added{(For a detailed treatment, see \cite{baier2008principles}).}  In Section \ref{sec:Markov_models}, we introduce Markov processes and define the labeling of a path on these processes and the probability measure associated with a set of paths. Next, we introduce \ltl and the finite-trace version \ltlf in Section \ref{sec:ltl}. We then define when a path on an MDP \deleted{to} satisfies an \ltl or \ltlf formula as well as the probability of satisfaction in Section \ref{sec:ltl_sat}. Finally, we give our formal problem definition in Section \ref{sec:problem_statement}.



\subsection{Markov Processes}
\label{sec:Markov_models}
Markov processes are frequently used to model systems that exhibit stochastic behavior.
While this paper deals with MDPs, it is also important to define \textit{Discrete-Time Markov Chains} (DTMCs) as they are needed to define probability measures.

\begin{definition}[DTMC]
	\label{def:DTMC}
	A labeled \emph{Discrete-Time Markov Chain} (DTMC) is a tuple $\mathcal{D} = (S, P, s_{\init}, AP, \L)$, where $S$ is a countable set of states, $s_{\init} \in S$ is the initial state, $AP$ is a finite set of atomic propositions, $\L: S \to 2^{AP}$ is a labeling function, and $P: S \times S \to [0,1]$ is a transition probability function where $\sum_{s'\in S} P(s,s') = 1$ for all $s \in S$.

\end{definition}






An execution of a DTMC is given by a path as defined below.
\begin{definition}[Path]
	A \emph{path} $\mypath$ through a DTMC is an infinite sequence of states:
	\begin{equation*}
		\mypath = s_0 s_1 \ldots s_n \ldots,
	\text{ \hspace{1em} such that \hspace{1em} }
		 P(s_i, s_{i+1}) > 0 \quad \forall i \geq 0.
	\end{equation*}
\end{definition}

\noindent
$Paths(s)$ is the set of paths starting in $s$. A finite path is \replaced{a path}{one} with a last state $s_n$. $Paths_{fin}(s)$ is the set of all finite paths starting in $s$. For every path $\mypath \in Paths(s)$, $pre(\mypath)$ denotes the set of prefixes of $\mypath$ (similarly for $\mypath \in Paths_{fin}(s)$).

A trace or labeling of a path is the sequence of labels of the states in the path.
\begin{definition}[Labeling of a Path]
	A \emph{labeling} (also referred to as \emph{valuation} or \emph{trace}) of an infinite path $\mypath = s_0 s_1 \ldots s_n \ldots$ is the sequence $\L(s_0) \L(s_1) \ldots \L(s_n) \ldots$
	The labeling of a finite path is defined analogously. We use $\L(\mypath)$ to denote the labeling of $\mypath$.
\end{definition}

To reason over the paths of a DTMC probabilistically, we need to define a probability space with a probability measure over infinite paths.  To this end, we first define cylinder sets that extend a finite path to a set of infinite paths.
\begin{definition}[Cylinder Set for DTMC]
	The cylinder set for some finite path $\mypath \in Paths_{fin}(s_{\init})$, denoted by $Cyl(\mypath)$, is the set of all infinite paths that share $\mypath$ as a prefix:
	\begin{equation}
		Cyl(\mypath) = \{ \mypath' \in Paths(s_{\init}) \mid \mypath \in pre(\mypath') \}.
	\end{equation}
\end{definition}

\begin{definition}[Probability Measure over Paths of a DTMC]
\label{def:probability_of_paths_dtmc}
For the probability space $(\Omega, \mathcal{E}, Pr)$, where sample space $\Omega=Paths(s_{\init})$, event space $\mathcal{E}$ is the smallest $\sigma$-algebra on $\Omega$ containing the cylinder sets of all finite paths (i.e., $Cyl(\mypath) \in \mathcal{E}$ for all $\mypath \in Paths_{fin}(s_{\init})$), the probability measure $Pr$ is defined as:
	\begin{equation}
		Pr(Cyl(s_0 \ldots s_n)) = \prod_{0 \leq i < n}P(s_i, s_{i+1}),
	\end{equation}
	where $s_0 = s_{\init}$. \added{This probability measure is unique per Carath\'eodory's extension theorem}. We define $Pr(Cyl(s_{\init})) = 1$. 
\end{definition}


Some systems exhibit not only probabilistic behavior but also non-deterministic behavior. These systems are typically modeled as Markov Decision Processes (MDPs). MDPs extend the definition of DTMCs by allowing choices of actions at each state.

\begin{definition}[MDP]
	\label{def:MDP}
	A labeled \emph{Markov Decision Process} (MDP) is a tuple: $\mathcal{M} = (S, \Act, P, s_{\init}, AP, \L)$, where $s_{init}$, $AP$ and $\L$ are as in Def.~\ref{def:DTMC} and:
	\begin{itemize}
    	\item $S$ is a finite set of states;
		\item $\Act$ is a finite set of actions, and $\Act(s) \subseteq \Act$ denotes the set of actions enabled at state $s \in S$;
		\item $P: S \times \Act \times S \to [0,1]$ is the transition probability function where $\sum_{s' \in S} P(s,a,s')=1$ for all $s\in S$ and $a \in \Act(s)$.
	\end{itemize}
\end{definition}

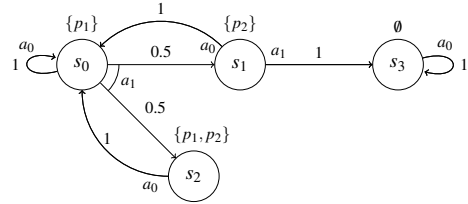
\begin{wrapfigure}{r}{0.4\textwidth}
\vspace{-5mm}
	\centering
	\scalebox{0.7}{
	\begin{tikzpicture}[auto,node distance=3cm,on grid, auto]
		\tikzstyle{round}=[thick,draw=black,circle]

		\node[state, label={\footnotesize $\{ p_1 \}$}] (s0) {$s_0$};
		\node[state, label={\footnotesize $\{ p_2 \}$}] (s1) [right=of s0] {$s_1$};
		\node[state, label={\footnotesize $\quad\{ p_1, p_2 \}$}] (s2) [below right=of s0] {$s_2$};
		\node[state, label={\footnotesize $\emptyset$}] (s3) [right=of s1] {$s_3$};		

		\path[->] 	(s0) edge node {\footnotesize 0.5} (s1)
					(s0) edge node {\footnotesize 0.5} (s2)
					(s0) edge [loop left] node {\footnotesize 1} ()
    				(s0) edge [loop left] node [above=1mm] {\footnotesize $a_0$} ();
		\path[->] 	(s1) edge [bend right=45] node [above] {\footnotesize 1} (s0)
					(s1) edge [bend right=45] node [very near start, below=0mm] {\footnotesize $a_0$} (s0)
				 	(s1) edge node [above] {\footnotesize 1} (s3)
				 	(s1) edge node [very near start, above] {\footnotesize $a_1$} (s3);
		\path[->] 	(s2) edge [bend left=45] node [above] {\footnotesize 1} (s0)
				 	(s2) edge [bend left=45] node [very near start, below=0mm] {\footnotesize $a_0$} (s0);
		\path[->] 	(s3) edge [loop right] node {\footnotesize 1} ()
					(s3) edge [loop right] node [near start, above] {\footnotesize $a_0$} ();

		\path pic[draw, angle radius=7mm,"\vphantom{g} \footnotesize{$a_1$}",angle eccentricity=1.3] {angle = s2--s0--s1};
		
	\end{tikzpicture}
	}
	\caption{Example MDP.}
  \label{fig:exampleMDP}
\end{wrapfigure}

\textit{Example 1.} An example of an MDP is \replaced{shown}{show} in \autoref{fig:exampleMDP}.  Actions ($A = \{a_0,a_1\}$) and probabilities are shown as edge labels. State names are within each state and state labels are above each state.


The notion of path can be straightforwardly extended from DTMCs to MDPs.

\begin{definition}[Path through MDP]
	A \emph{path} $\mypath$ through an MDP is a state followed by a sequence of action-state pairs:
	\begin{equation*}
		\mypath = s_0 \langle a_0, s_1 \rangle \langle a_1, s_2 \rangle \ldots \langle a_{n-1}, s_n \rangle \ldots
	\end{equation*}
	such that
	$s_0 = s_\init$, $a_i \in \Act(s_i)$, and $P(s_i, a_i, s_{i+1}) > 0$ for all $i\geq 0$.
	$Paths(s)$, $Paths_{fin}(s)$ and $pre(\mypath)$ are defined analogously to the DTMC case.

\end{definition}

A \emph{policy} (also known as an \emph{adversary} or \emph{strategy}) resolves the choice of actions. This induces a (possibly infinite) DTMC.

\begin{definition}[Policy]
	A \emph{policy} $ \pi: Paths_{fin} \to \Act$ maps every finite path $\mypath = s_0 \langle a_0,s_1 \rangle \langle a_1, s_2 \rangle \ldots \langle a_{n-1}, s_n \rangle $ to an element of $\Act(s_n)$, i.e., $\pi(\mypath) \in \Act(s_n)$.\footnote{Here, we focus on  deterministic policies as they are sufficient for optimality of \ltl (and \ltlf) properties on MDPs \cite{baier2008principles}.}
\end{definition}

\noindent
The set of all policies is $\Pi$. A policy is \emph{stationary} if $\pi(\mypath)$ only depends on the most recent state $s_n$ of $\mypath$ denoted by $last(\mypath)$. Under policy $\pi$ the action choices on $\M$ are determined. This gives us a DTMC whose states correspond to the finite paths of the MDP. We call this the DTMC induced on $\mathcal{M}$ by policy $\pi$ and denote this by $\mathcal{D}^\pi$.  $Paths^\pi(s)$ is shorthand for $Paths(s)$ of $\mathcal{D}^\pi$. 
Therefore, the probability of the paths of the MDP under $\pi$ are defined according to $\mathcal{D}^\pi$, whose probability space includes a probability measure $Pr^\pi$ over infinite paths via cylinder sets that extend a finite path to a set of infinite paths. See \cite{baier2008principles} for details.


\subsection{Linear Temporal Logic}
\label{sec:ltl}


Linear temporal logic (\ltl) is a popular formalism used to specify temporal properties. Here, we are interested in \ltl interpreted over finite traces, but we must still define \ltl as we use a reduction to \ltl as one of our two approaches to synthesis.

\begin{definition}[\ltl Syntax]
	An $\ltl$ formula is built from a set of propositional symbols $Prop$ and is closed under the boolean connectives as well as the ``next'' operator $X$ and the ``until'' operator $U$: $$\ltlphi ::= \top   \mid   p   \mid   (\neg \ltlphi)   \mid   (\ltlphi_1 \wedge \ltlphi_2)   \mid   (X \ltlphi)  \mid  (\ltlphi_1 U \ltlphi_2),$$ where $p \in Prop.$  
\end{definition}
The common temporal operators ``eventually'' ($F$) and ``globally'' ($G$) are defined as: $F \, \ltlphi = \top \, U \, \ltlphi$ and $G \, \ltlphi = \neg F\, \neg \ltlphi$.
The semantics of \ltl are defined over infinite traces (for full definition, see \cite{pnueli1977temporal}).  An \ltl formula $\ltlphi$ defines an $\omega$-regular language $\mathcal{L}(\ltlphi)$ over alphabet $2^{Prop}$.
\deleted{i.e., $\mathcal{L(\ltlphi)} = \{\rho \in (2^{Prop})^\omega \mid \rho \models \ltlphi\}.$}

Next we define \ltlf.  To distinguish between the formulas of the two logics, we use $\ltlphi$ for \ltl  and $\ltlfphi$ for \ltlf formulas.

\begin{definition}[\ltlf Syntax \& Semantics]
	An \ltlf formula has identical syntax to \ltl, but the semantics is defined over finite traces. 
	Let $|\rho|$ and $\rho_i$ denote the length of trace $\rho$ and the symbol in the $i^{th}$ position in $\rho$, respectively, and $\rho, i \models \ltlfphi$ is read as: ``the $i^{th}$ step of trace $\rho$ is a model of $\ltlfphi$.'' 
	Then,
	\begin{itemize}
		\item $\rho , i \models \top;$
		\item $\rho , i \models p$ iff $p \in \rho_i;$
		\item $\rho , i \models \neg \ltlfphi$ iff $\rho, i \not \models \ltlfphi;$
		\item $\rho , i \models \ltlfphi_1 \wedge \ltlfphi_2$ iff $\rho, i \models \ltlfphi_1$ and $\rho, i \models \ltlfphi_2;$
		\item $\rho , i \models X  \ltlfphi$ iff $|\rho| > i+1$ and $\rho, i+1 \models \ltlfphi;$
		\item $\rho , i \models \ltlfphi_1 U \ltlfphi_2$ iff $\exists j$ s.t. $ i \leq j < |\rho|$ and $\rho, j \models \ltlfphi_2$ and $\forall k$, $i \leq k < j$, \; $\rho, k \models \ltlfphi_1$.
	\end{itemize}
\end{definition}
\noindent
We say finite trace $\rho$ satisfies formula $\ltlfphi$, denoted by $\rho \models \ltlfphi$, iff $\rho,0 \models \ltlfphi$.
An \ltlf formula $\ltlfphi$ defines a language $\mathcal{L}(\ltlfphi)$ over the alphabet $2^{Prop}$. $\mathcal{L}(\ltlfphi)$ is a star-free regular language \cite{de2013linear}.


\subsection{Satisfaction of Temporal Logic Specification}
\label{sec:ltl_sat}

Here, we define what it means for an MDP to satisfy an \ltl or \ltlf formula.  

\begin{definition}[Path satisfying \ltl]
	For a pair ($\mathcal{M}$, $\ltlphi$) of an MDP and an \ltl formula where the atomic propositions of $\mathcal{M}$ match the propositions of $\ltlphi$ (i.e., $Prop = AP$), we say that an infinite path $\mypath$ on $\mathcal{M}$ satisfies specification $\ltlphi$ if the labeling of $\mypath$ is in the language of $\ltlphi$, i.e., $\L(\mypath) \in \mathcal{L}(\ltlphi)$.
\end{definition}

\noindent
Following \cite{zhu2017symbolic}, we define finite satisfaction (of an \ltlf formula) as follows.

\begin{definition}[Path satisfying \ltlf]
	For a pair ($\mathcal{M}$, $\ltlfphi$) of an MDP and an \ltlf formula where the atomic propositions of $\mathcal{M}$ match the propositions of $\ltlfphi$ (i.e., $Prop = AP$), we say that a (possibly finite) path $\mypath$ of $\mathcal{M}$ satisfies specification $\ltlfphi$ if at least one prefix of $\mypath$ is in the language of $\ltlfphi$, i.e., 
	\begin{equation}
		\mypath \models \ltlfphi \quad \Leftrightarrow \quad \exists \mypath' \in pre(\mypath) \; s.t. \; \L(\mypath') \in \mathcal{L}(\ltlfphi).
	\end{equation}
\end{definition}

\noindent
Intuitively, this corresponds to a system that can declare its execution complete after satisfying its goals. \ltlf is suitable for specifications that are to be completed in finite time.

\begin{definition}[Probability of \ltlf satisfaction]
	The probability of satisfying an \ltlf property $\ltlfphi$ in $\mathcal{M}$ under policy $\pi$ is $Pr(\mathcal{M}^\pi \models \ltlfphi) = Pr^\pi(\mypath \in Paths^\pi(s_{\init}) \mid \mypath \models \ltlfphi)$.
\end{definition}

\subsection{Problem Statement}
\label{sec:problem_statement}
We formalize the problem of \ltlf synthesis on MDPs as:
\begin{problem}[\ltlf synthesis on MDPs]
\label{problem}
	Given MDP $\mathcal{M}$ and an \ltlf formula $\ltlfphi$, compute a policy $\pi^*$ that maximizes the probability of satisfying $\ltlfphi$, i.e.,
$$\pi^* = \arg \max_{\pi \in \Pi} Pr(\mathcal{M}^\pi \models \ltlfphi),$$	
as well as this probability, i.e., $Pr(\M^{\pi^*} \models \ltlfphi)$.
\end{problem}


\section{Synthesis Algorithms}
\label{sec:synthesis_alg}

We introduce two approaches to \ltlf policy synthesis on MDPs.  The first approach is based on reduction to classical \ltl policy synthesis on MDPs.  The second approach is through a translation of \ltlf formulas to first order logic formulas, which can be translated to a symbolic deterministic automaton.  In this section, we detail these two algorithms, and in Section~\ref{sec:evaluation}, we show that the second approach scales better than the classical \ltl approach.

\subsection{Reduction to \ltl Synthesis}
\label{sec:ltlf2ltl_synthesis}

We can reduce the problem of \ltlf policy synthesis on MDPs to a classical \ltl policy synthesis.  The algorithm consists of two main steps: (1) construction of an MDP $\M'$ from $\M$ by augmenting it with an additional state and atomic proposition, and (2) translation of \ltlf formula  $\ltlfphi$ on the labels of $\M$ to its equivalent \ltl formula $\ltlphi$ on the labels of $\M'$.

\subsubsection{Augmented MDP}
\label{sec:augMDP}

Recall that the semantics of \ltlf formulas is over finite traces whereas the interpretation of \ltl formulas is over infinite traces.  In order to reduce the \ltlf synthesis problem to an \ltl one, we need to 
be able to capture the finite paths (traces) of $\M$ that satisfy $\ltlfphi$ and extend them to infinite paths (traces).  Specifically, we need those satisfying finite paths that contain no $\ltlfphi$-satisfying proper prefixes, i.e., satisfy $\ltlfphi$ for the first time.  To this end, we allow the \replaced{system}{environment} (policy) to decide when to ``terminate.''  We view the system to be ``alive'' until termination, at which point it is no longer alive.  Then, we define an \ltl formula that requires the system to be alive while it has not satisfied $\ltlfphi$ and to terminate after satisfying $\ltlfphi$.


To this end, we augment MDP $\M$ with a terminal action and state and an atomic proposition $alive$.  Formally, we construct MDP $\M'=(S', \Act', P', s'_{\init}, AP', \L')$, where $S' = S \cup \{s_\term\}$, $\Act' = \Act \cup \{a_\term\}$, $s'_\init = s_\init$, $AP' = AP \cup \{alive\}$,
\begin{equation}
    \label{eq:aug_mdp}
    \begin{split}
    	\Act'(s) = 
    		\begin{cases}
    			\Act(s) \cup \{a_\term\}	&	\text{if } s \neq s_\term\\
    			\{a_\term\}					&	\text{if } s = s_\term
    		\end{cases}
    		, \quad \quad
    			\L'(s) = 
    		\begin{cases}
    			\L(s) \cup \{alive\}	&	\text{if } s \neq s_\term \\
    			\emptyset				&	\text{if } s = s_\term
    		\end{cases}
    		,\\
    	P'(s,a,s') = 
    		\begin{cases}
    			P(s,a,s')	&	\text{if } s \in S, \: a \in \Act, s' \in S \\
    			0			&	\text{if } s \in S, \: a \in \Act, s' = s_\term\\
    			1			&	\text{if } s \in S', a=a_\term, s' = s_\term
    		\end{cases}
    		. \qquad\qquad
    \end{split}
\end{equation}

In this MDP, the system can decide to terminate by taking action $a_\term$, in which case it transitions to state $s_\term$ with probability one and remains there forever.  The labeling of the corresponding path includes the atomic proposition $alive$ at every time step until $s_\term$ is visited and is empty thereafter.

\begin{wrapfigure}{r}{0.55\textwidth}
    \vspace{-5mm}
	\centering
	\scalebox{0.78}{
	\begin{tikzpicture}[auto,node distance=3cm,on grid, auto]
		\tikzstyle{round}=[thick,draw=black,circle]

		\node[state, label={\footnotesize $\{ p_1, alive \}$}] (s0) {$s_0$};
		\node[state, label={\footnotesize $\quad \{ p_2, alive \}$}] (s1) [right=of s0] {$s_1$};
		\node[state, label={\footnotesize $\quad\quad\quad\{ p_1, p_2, alive \}$}] (s2) [below right=of s0] {$s_2$};
		\node[state, label={\footnotesize $\{alive\}$}] (s3) [right=of s1] {$s_3$};		
		\node[state, label={$\emptyset$}] (sterm) [right=of s2] {$s_\term$};		

		\path[->] 	(s0) edge node {\footnotesize 0.5} (s1)
					(s0) edge node {\footnotesize 0.5} (s2)
					(s0) edge [loop left] node {\footnotesize 1} ()
    				(s0) edge [loop left] node [above=1mm] {\footnotesize $a_0$} ()
    				(s0) edge [dashed, bend right=100] node [very near start, left=0mm] {\footnotesize $a_\term$} (sterm);
		\path[->] 	(s1) edge [bend right=45] node [above] {\footnotesize 1} (s0)
					(s1) edge [bend right=45] node [very near start, below=0mm] {\footnotesize $a_0$} (s0)
				 	(s1) edge node [above] {\footnotesize 1} (s3)
				 	(s1) edge node [very near start, above] {\footnotesize $a_1$} (s3)
				 	(s1) edge [dashed] node [near start, right=1.5mm] {\footnotesize $a_\term$} (sterm);
		\path[->] 	(s2) edge [bend left=45] node [above] {\footnotesize 1} (s0)
				 	(s2) edge [bend left=45] node [very near start, below=0mm] {\footnotesize $a_0$} (s0)
					(s2) edge [dashed] node [near start, below=0mm] {\footnotesize $a_\term$} (sterm);
		\path[->] 	(s3) edge [loop right] node {\footnotesize 1} ()
					(s3) edge [loop right] node [near start, above] {\footnotesize $a_0$} ()
					(s3) edge [dashed] node [near start, right] {\footnotesize $a_\term$} (sterm);
		\path[->]	(sterm) edge [dashed, loop right] node [above=1mm] {\footnotesize $a_\term$} ();

		\path pic[draw, angle radius=7mm,"\vphantom{g} \footnotesize{$a_1$}",angle eccentricity=1.3] {angle = s2--s0--s1};
		
	\end{tikzpicture}
	}
	\vspace{-13mm}
	\caption{Augmented MDP constructed from $\M$ in Fig.~\ref{fig:exampleMDP}. 
	}
	\vspace{-4mm}
  \label{fig:augmentedMDP}

\vspace{-5mm}
\end{wrapfigure}
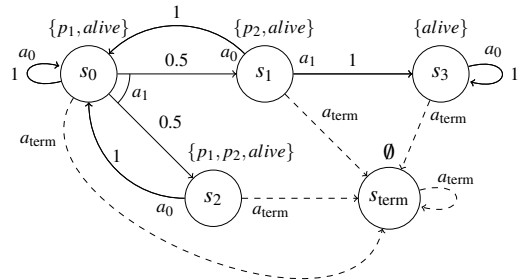

\textit{Example 2.} \autoref{fig:augmentedMDP} illustrates the augmented MDP $\M'$ constructed from the example MDP $\M$ in \autoref{fig:exampleMDP}.  State $s_\term$ along with the dashed edges and atomic proposition $alive$ are added to $\M$.  The dashed edges are enabled by action $a_\term$ and have transition probability of one. The label of $s_\term$ is the empty set, and the labels of the rest of the states contain $alive$.



With this augmentation, the system is able to terminate once a satisfying finite path is generated. We show that, for finite properties, the maximum probability of satisfaction in $\M$ equals the maximum probability of satisfaction in $\M'$ under a policy that enforces a visit to $s_\term$. 


Let $Paths_{fin,\ltlfphi} \subseteq Paths_{fin}(s_{\init})$ be a set of finite paths of interest in \added{arbitrary} MDP $\M$.  Denote the probability of satisfaction of infinite paths with prefixes from $Paths_{fin, \ltlfphi}$ under policy $\pi \in \Pi$ by:

\begin{equation}
	\label{eq:prob_interested_paths}
	Pr(\M^\pi \models Paths_{fin,\ltlfphi}) =   
	Pr^\pi (\mypath \in Paths^\pi (s_{\init}) \mid  pre(\mypath) \cap Paths_{fin,\ltlfphi} \neq \emptyset).
\end{equation}
The following lemma states  the equivalence of the maximum probability of satisfying this set of paths in $\M$ and $\M'$.

\begin{lemma}
    \label{lemma:pathProb}
    Let $Paths_{fin,\ltlfphi} \subseteq Paths_{fin}(s_{\init})$ be a set of finite paths of interest of MDP $\M$.
    Further, let $\M'$ be the augmented version of $\M$ and $\Pi'$ be the set of policies of $\M'$ that enforce a visit to state $s_\term$ (take action $a_{\term}$). Then, it holds that
    \begin{equation}
        \max_{\pi \in \Pi} Pr(\M^\pi \models Paths_{fin,\ltlfphi}) = \max_{\pi' \in \Pi'} Pr(\M'^{\pi'} \models Paths_{fin,\ltlfphi}).
    \end{equation}
\end{lemma}

\begin{proof}

We prove the results by showing that LHS $\leq$ RHS and then LHS $\geq$ RHS.

\textit{Case 1.} We want to prove that: $$\max_{\pi \in \Pi} Pr(\M^\pi \models Paths_{fin,\ltlfphi}) \leq \max_{\pi' \in \Pi'} Pr(\M'^{\pi'} \models Paths_{fin,\ltlfphi}).$$
Consider policy $\pi \in \Pi$ on $\M$. We can construct policy $\pi'$ on $\M'$ from $\pi$ as follows:
\begin{equation*}
	\label{eq:aug_policy}
	\pi'(\mypath_{fin}) = 
	\begin{cases}
		\pi(\mypath_{fin}) 	& \text{if } \mypath_{fin} \notin Paths_{fin,\ltlfphi} \; \wedge \; last(\mypath_{fin})\neq s_\term \\
		a_\term 			& \text{otherwise}.
	\end{cases}
\end{equation*}
Recall the probability measure given by the cylinder set $Pr(Cyl(s_0 \ldots s_n)) = \prod_{0 \leq i < n}P(s_i, \pi^*(s_i), s_{i+1})$ and that the transition probability under $a_\term$ is always $1$. Thus from optimal policy $\pi^*$ on $\M$, we can construct $\pi'$ on $\M'$ with equal probability.  

\textit{Case 2.}
Next we prove that: $$\max_{\pi \in \Pi} Pr(\M^\pi \models Paths_{fin,\ltlfphi}) \geq \max_{\pi' \in \Pi'} Pr(\M'^{\pi'} \models Paths_{fin,\ltlfphi}),$$

\noindent
Consider an optimal policy $\pi'^*$ on $\M'$. By assumption, we know that $\pi'^*$ chooses $a_{\term}$ at some finite point along any path. Consider an arbitrary path and let $a_{\term}$ be chosen as the $n+1^{th}$ action: $\mypath'^* = s_0 \ldots s_n s_{\term} \ldots$ Let $k \leq n$ be the smallest natural number such that $\mypath' = s_0 \ldots s_k \in Paths_{fin,\ltlfphi}$. Then consider policy $\pi'$ which takes action $a_{\term}$ at state $s_k$ and with probability 1 transitioned to state $s_{\term}$ ($\pi'$ and $\pi'^*$ may be identical). Recall the probability measure given by the cylinder set $Pr(Cyl(s_0 \ldots s_n)) = \prod_{0 \leq i < n}P(s_i, \pi'^*(s_i), s_{i+1})$. Thus, because $a_{\term}$ has probability 1, $Pr(\M'^{\pi'} \models Paths_{fin,\ltlfphi}) \geq Pr(\M'^{\pi'^*} \models Paths_{fin,\ltlfphi})$. We can map $\pi'$ to a policy $\pi$ for $\M$ with the same probability of satisfaction by changing every choice of $a_{\term}$ to an arbitrary available action. Thus, from (optimal policy) $\pi'$ we can construct a policy $\pi$ on $\M$ with an equal probability of satisfaction.

\end{proof}

\subsubsection{\ltlf to \ltl}

To translate an \ltlf formula on $\M$ to its equivalent \ltl formula on $\M'$, we follow \cite{de2013linear}.
Let $\Phi_f$ be the set of \ltlf formulas $\ltlfphi$ defined over atomic propositions in $AP$ and $\Phi$ be the set of \ltl formulas $\ltlphi$ defined over $AP'$.  Then 
$g : \Phi_f \rightarrow \Phi$
is defined as:
\begin{equation}
    \label{eq:ltlf2ltl}
    g(\ltlfphi) = t(\ltlfphi) \, \wedge \, \big(alive \: U \: (G \: \neg alive)\big),
\end{equation}
where $t : \Phi_f \rightarrow \Phi$
is inductively defined as:
\begin{itemize}
	\item $t(p) = (p \wedge alive)$, where $ p \in AP$;
	\item $t(\neg \ltlfphi) = \neg t (\ltlfphi)$;
	\item $t(\ltlfphi_1 \wedge \ltlfphi_2) = t(\ltlfphi_1) \wedge t(\ltlfphi_2)$;
	\item $t(X \: \ltlfphi) = X \: (alive \: \wedge \:  t(\ltlfphi))$;
	\item $t(\ltlfphi_1 \: U  \: \phi_2) = t(\ltlfphi_1) \: U \: (alive \: \wedge \: t(\ltlfphi_2))$.
\end{itemize}

\noindent
In this construction, mapping $t$ ensures that $alive$ is present in the last letter of every finite trace that satisfies $\ltlfphi$.  Then, $g$ translates $\ltlfphi$ to $\ltlphi$ by requiring $alive$ to be true until $\ltlfphi$ is satisfied and false thereafter.  In other words, the translated \ltl formula $\ltlphi = g(\ltlfphi)$ requires the system to terminate after it satisfies $\ltlfphi$.

\begin{theorem}
	\label{theorem:ltlsynthesis}
	Given an MDP $\M$ and its augmentation $\M'$ as well as an \ltlf formula $\ltlfphi$ and the corresponding \ltl formula $g(\ltlfphi)$, then 
	$$\max_{\pi \in \Pi} Pr(\M \models \ltlfphi) = \max_{\pi' \in \Pi'} Pr(\M' \models g(\ltlfphi)),$$ 
	where $\Pi'$ is the set of all policies in $\M'$.
	
	
\end{theorem}

\begin{proof}
This follows through an application of Lemma \ref{lemma:pathProb} and a proof that the finite paths on $\M$ satisfying $\ltlfphi$ correspond to the infinite paths on $\M'$ satisfying $g(\ltlfphi)$. First we note that the termination assumption of Lemma \ref{lemma:pathProb} (i.e., that policies visit $s_{\term}$) is met because $g(\ltlfphi)$ contains the clause $\big(alive \: U \: (G \: \neg alive)\big)$. For all $s \in \M$ and $s' \in \M' \setminus \{ s_{\term} \}: L(s) = L(s') \setminus \{ alive \}$. Recall that $s_{\term}$ is a sink state with label $\emptyset$. Then from \cite{de2013linear} it follows that finite paths of $\M$ that satisfy $\ltlfphi$ correspond to paths in $\M'$ that satisfy $g(\ltlfphi)$.




\end{proof}



A direct result of the above theorem is the reduction of \ltlf policy synthesis to classical \ltl policy synthesis on MDPs.
\begin{corollary}
	\label{corollary:LTLreduction}
	The policy synthesis to maximize the probability of satisfying \ltlf formula $\ltlfphi$ on $\M$ can be reduced to the \ltl maximal policy synthesis problem.
\end{corollary}

Therefore, to solve Problem~\ref{problem}, we can use $\ltl$ synthesis on augmented MDP $\M'$ and property $\ltlphi = g(\ltlfphi)$.  The general \ltl synthesis algorithm is well-established \cite{baier2008principles} and follows the following steps: (1) translation of the \ltl formula to a DRA, (2) composition of the DRA with the MDP, which results in another MDP called the product MDP, (3) identification of the maximal end-components on the product MDP that satisfy the accepting condition of DRA, and finally (4) solving the maximal reachability probability problem (stochastic shortest path problem) \cite{de1997formal} on the product MDP with the accepting end-components as the target states.  There exist many tools such as PRISM \cite{kwiatowska2011prism} that can solve the \ltl synthesis problem.  Specifically, PRISM has a symbolic implementation of this algorithm, enabling fast computations.  In Section \ref{sec:evaluation}, however, we show that the native approach introduced below outperforms the \ltl-reduction approach even by using PRISM's symbolic engine.

\vspace{-1mm}
\subsection{Native Approach}

\label{sec:native_synthesis}

In the native approach, we first convert the \ltlf formula into a \textit{deterministic finite automaton} (DFA) using specialized tools \cite{zhu2017symbolic}. Then we take the product of this automaton with the MDP. Finally, we synthesize a strategy by solving the maximal reachability probability problem on this product MDP. 

\vspace{-1mm}
\subsubsection{Translation to DFA}
\vspace{-1mm}

A Deterministic Finite Automaton (DFA) is a tuple: $\mathcal{A} = (Q, \Sigma, \delta, q_0, F)$, where $Q$ is the set of states, $\Sigma$ the alphabet, $\delta: Q \times \Sigma \mapsto Q$ the transition function, $q_0$ the initial state and $F$ the set of accept states.
A finite \textit{run} of a DFA on a trace $\trace = \trace_0 \trace_1 \ldots \trace_n$, where $\trace_i \in \Sigma$, is the sequence of states 
$q_0 q_1 \ldots q_{n+1}$ such that $q_{i+1} = \delta (q_i, \trace_i)$ for all $0 \leq i \leq n$.  This run is accepting if $q_{n+1} \in F$.  

Following \cite{zhu2017symbolic}, we translate \ltlf to a DFA using MONA \cite{henriksen1995mona}.
\ltlf is expressively equivalent to First-order Logic on finite words, which in turn is a fragment of Weak Second-order Theory of One Successor. We use the translation given in \cite{de2013linear} to convert \ltlf formula $\ltlfphi$ to a First-order Logic formula. MONA offers translations from Weak Second-order Theory of One or Two Successors (WS1S/WS2S) to a DFA that accepts precisely the language of our \ltlf formula $\ltlfphi$.  We denote this DFA by $\A_\ltlfphi$.

\vspace{-1mm}
\subsubsection{Product of DFA with MDP}
\vspace{-1mm}

Given DFA $\A_\ltlfphi$, we can take the product with the MDP $\M$ to achieve a new MDP $M \times \A_\ltlfphi$ as follows.
The product of MDP $\M = (S, \Act, P, s_{\init}, AP, \L)$ and DFA $\A_\ltlfphi = (Q, \Sigma, q_0, \delta, F)$ is an MDP 
\[\M \times \A_\ltlfphi = (S \times Q, \Act, P^{\M \times \A_\ltlfphi}, (s_{\init}, q_{\init})),\] 
where
$q_{\init} = \delta(q_0, L(s_{\init}))$, and
\begin{equation*}
	P^{\M \times \A_\ltlfphi}((s, q), a, (s', q')) =
	\begin{cases}
		P(s, a, s') & \text{if } q' = \delta(q, \L(s)) \\
		0             & \text{otherwise }
	\end{cases}.
\end{equation*}
The paths of this product MDP have one-to-one correspondence to the paths of MDP $\M$ as well as the runs of $\A_\ltlfphi$.  Therefore, the projection of the paths of $\M \times \A_\ltlfphi$ that reach state $(s,q)$, where $q \in F$, on $\A_\ltlfphi$ are accepting runs and on $\M$ are $\ltlfphi$-satisfy paths.  Thus, we solve Problem~\ref{problem} by synthesizing an optimal policy on this product MDP,  using standard tools for the maximal reachability probability problem as discussed in Sec.~\ref{sec:ltlf2ltl_synthesis}.  The resulting policy is stationary on the product MDP but history-dependent on $\M$.

\added{In theory, both the \ltl approach and the native approach exhibit runtimes doubly-exponential in the size of the formula, due to the need to construct deterministic automata \cite{ Courcoubetis1995complexity, de2015synthesis}. In practice, the native approach's translation to a DFA using specialized tools offers better runtime and memory usage compared to the \ltl-based approach.} Additionally, it produces a minimal DFA, while the \ltl pipeline produces an $\omega$-automaton which cannot be minimized effectively by existing tools. We show the benefits of the native approach experimentally in Sec.~\ref{sec:evaluation}.


\section{Evaluation}
\label{sec:evaluation}

We evaluate the proposed synthesis approaches through a series of benchmarking case studies.  Below, we provide details on our implementation, experimental scenarios, and obtained results. A version of our tool along with examples is available on GitHub \cite{wells2020github}.

\subsection{Experimental Framework}

We run our experiments using the PRISM framework \cite{kwiatowska2011prism}. PRISM uses a symbolic encoding of MDPs as well as an encoding of automata as a list of edges, where the labels of edges are encoded symbolically using BDDs. PRISM supports several tools for the \ltl-to-automata translation. We tested PRISM's built-in translator as well as Rabinizer, LTL3DRA and SPOT \cite{kretinsky2018rabinizer,babiak2013ltl3dra,duret2016spot}. Of these, PRISM's built in conversion and SPOT performed significantly better than the others, and were used for evaluation.

We note that we also considered a comparison study against MoChiBa~\cite{sickert2016mochiba}, a tool based on PRISM that uses Limit-Deterministic B\"uchi Automata. 
Unfortunately, it is not possible to conduct this study since the current implementation of MoChiBa does not support loading automata from external tools and only runs with a modified implementation of PRISM's explicit engine. Our tests use PRISM's Hybrid engine (enabled by default). We leave a comparison of our DFA-based approach for \ltlf synthesis to an approach based on Limit-Deterministic B\"uchi Automata to future work. 


In the implementation of the \ltl-reduction approach, we augment the MDP and convert the \ltlf formula into an equivalent \ltl formula as described in Section \ref{sec:ltlf2ltl_synthesis}. Then, we input both the \ltl formula and the modified MDP into PRISM for synthesis.
%
%
In the implementation of the native approach, we invoke PRISM on the original \ltlf formula and MDP and use an external tool to convert the \ltlf formula into a DFA using MONA \cite{henriksen1995mona} then convert from MONA's format to the HOA format \cite{Babiak2015HOA}. Note that external tools (SPOT and our native approach) read from hard disk, whereas using PRISM's built-in conversion avoids this. Nevertheless, even including this time, the native approach (and sometimes SPOT) typically gives better performance as shown below. All experiments are run on a computer with an Intel i7-8550U and 16GB of RAM. PRISM is run using the default settings.

\vspace{-3mm}
\subsection{Experimental Scenarios}

\subsubsection{Test MDPs}
We consider four types of MDPs: \textit{Gridworld}, \textit{Dining Philosophers}, \textit{Nim}, and \textit{Double Counter}. 
In Gridworld, an agent is given some goal as an \ltlf formula and must maximize the probability of satisfaction. The agent has four actions: North, South, East, and West.  Under each action, the probability of moving to the cell in the intended direction is 0.69, then 0.01 for its opposite cell, and finally 0.1 for each of the other directions and for remaining in the same cell.  
If the resulting movement would place the agent in collision with the boundary, the agent remains in its current cell.  We model the motion of this agent as an MDP, where the states correspond to the cells of the grid and the set of actions is $\{a_\text{north}, a_\text{south}, a_\text{east}, a_\text{west}\}$.

We base the Dining-Philosophers domain on the tutorial from the PRISM website. We consider a ring of five philosophers with two different specifications. Note that typical Dining-Philosophers specifications of interest are infinite-trace; our tests are merely meant to show that our approach works on MDPs other than Gridworld. Both \textit{Nim} and \textit{Double Counter} are probabilistic versions of games presented in \cite{tabajara2019partitioning}.


\subsubsection{Test formulas}
Unfortunately, there is no standard set of \ltlf formulas that we can use for benchmarking. In \cite{zhu2017symbolic}, random \ltl formulas are used as the basis of \ltlf benchmarks; however, because we also consider an MDP, random formulas are frequently tautologies or non-realizable with respect to the MDP.  
For instance, consider a randomly generated formula $\ltlfphi = p_1 U p_2$, where $p_1$ and $p_2$ are randomly assigned to the labels of the states in the Gridworld MDP.  There is likely no path of the MDP that can satisfy this formula. As a matter of fact, in a test of more than 70 randomly generated \ltlf formulas and Gridworld MDPs, only one was ``interesting'' (yielding probability between zero and one) on the corresponding MDP. For benchmarking, the interest is in sets of formulas that make sense in a probabilistic setting.

Therefore, as the first set of test formulas for the \textit{Gridworld} MDP, we considered a natural finite-trace specification where the agent has $n$ goals to accomplish in any order, as well as a ``safety'' property where it must globally avoid some states. This is typical of e.g., a robotics domain \cite{he2019efficient}. We use \textbf{Fn} to denote this formula for a given $n$.
For the second set of test formulas, we keep regions to visit and avoid but introduce some ordering and repeat visits. \textbf{OS} is a short formula introducing order and \textbf{OL} is a longer formula that also contains nested temporal operators. All formulas are given in the appendix of the online version \cite{wells2019ltlf}.


For the \textit{Dining Philosophers} domains, typical examples focus on infinite-run properties. We test several finite-run properties on these domains. These properties are meant to illustrate our tool on an MDP other than Gridworld, but are not representative of typical finite-trace properties. One version (\textbf{D5}) asserts that all philosophers must eat at least once. The other (\textbf{D5C}) is a more complex property involving orders of eating. Both formulas are available in the appendix of the online version \cite{wells2019ltlf}.

For \textit{Nim}, the game of Nim is played against a stochastic opponent. 
To increase difficulty, the specifications require that randomly chosen stack heights are to be reached or avoided by the system player.
For \textit{Double Counter}, two four-bit binary counters are used.  One counter is controlled by a stochastic environment, and the other is controlled by the system. The aim of the system is to make its counter match that of the environment.

\begin{figure*}[t!]
    \centering
    \begin{subfigure}[b]{0.32\textwidth}
        \centering
        \includegraphics[trim=15 5 30 30, clip, width=\textwidth]{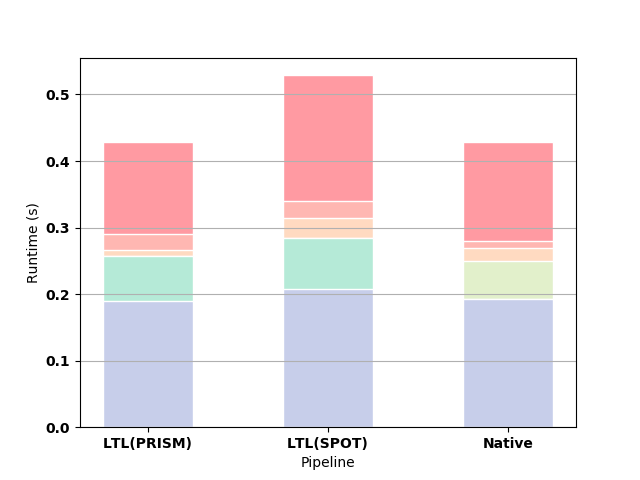}
        \vspace{-6mm}
        \caption{$10 \times 10$ grid with \textbf{F3}.}
        \label{fig:times_10x10_F3}
    \end{subfigure}
    \begin{subfigure}[b]{0.32\textwidth}
        \centering
        \includegraphics[trim=15 5 30 30, clip, width=\textwidth]{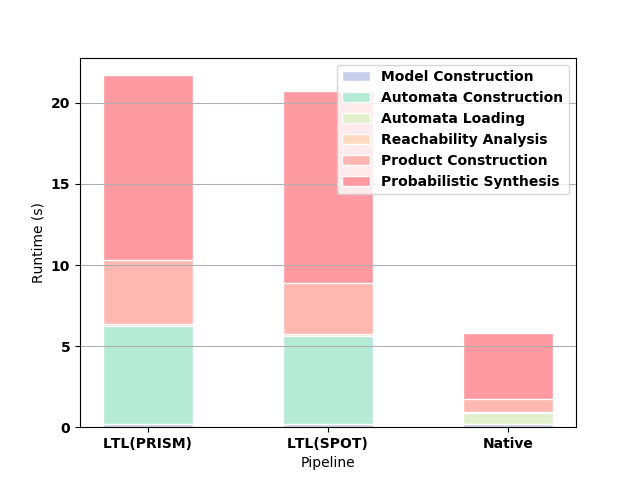}
        \vspace{-6mm}
        \caption{$10 \times 10$ grid with \textbf{F8}.}
        \label{fig:times_10x10_F8}
    \end{subfigure}
    \begin{subfigure}[b]{0.32\textwidth}
        \centering
        \includegraphics[trim=15 5 30 30, clip, width=\textwidth]{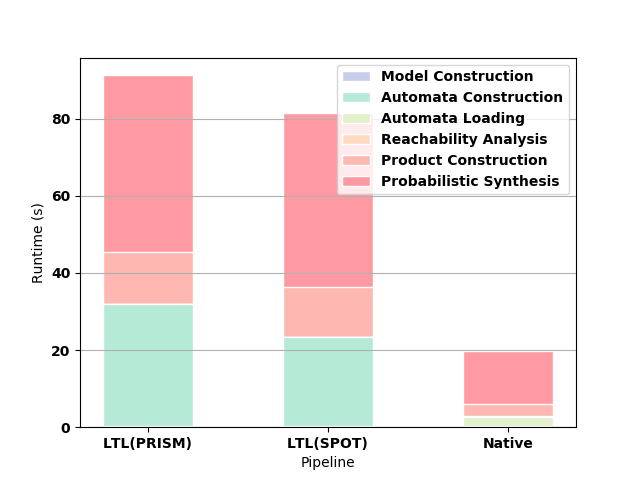}
        \vspace{-6mm}
        \caption{$10 \times 10$ grid with \textbf{F9}.}
        \label{fig:times_10x10_F9}
    \end{subfigure}
    
    \begin{subfigure}[b]{0.32\textwidth}
        \centering
        \includegraphics[trim=15 5 30 30, clip, width=\textwidth]{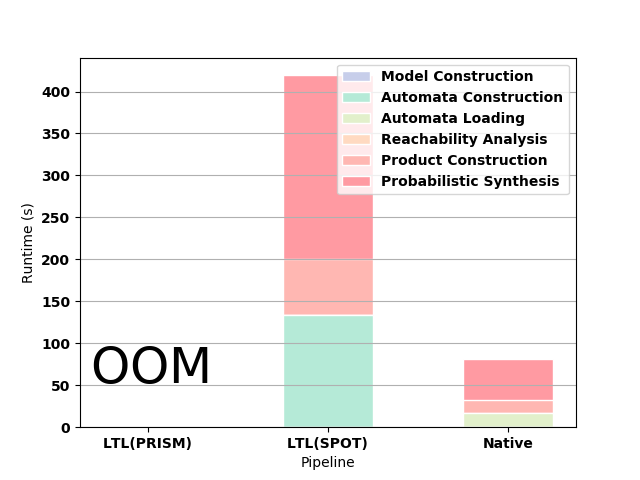}
        \vspace{-6mm}
        \caption{$10 \times 10$ grid with \textbf{F10}.}
        \label{fig:times_10x10_F10}
    \end{subfigure}
    \begin{subfigure}[b]{0.32\textwidth}
        \centering
        \includegraphics[trim=15 5 30 30, clip, width=\textwidth]{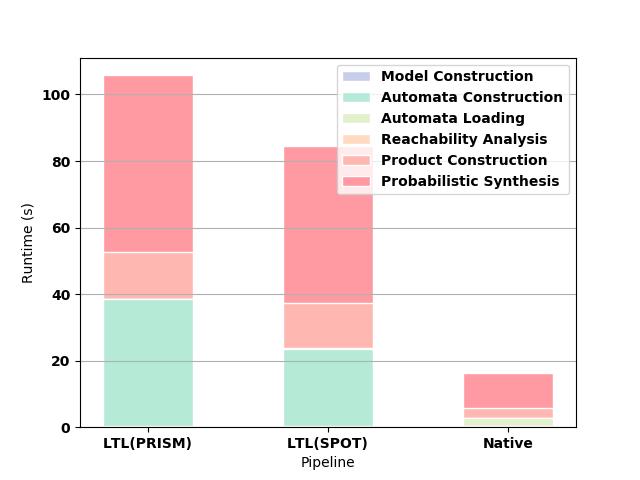}
        \vspace{-6mm}
        \caption{Random grid with \textbf{F9}.}
        \label{fig:times_random_F9}
    \end{subfigure}
    \begin{subfigure}[b]{0.32\textwidth}
        \centering
        \includegraphics[trim=15 5 30 30, clip, width=\textwidth]{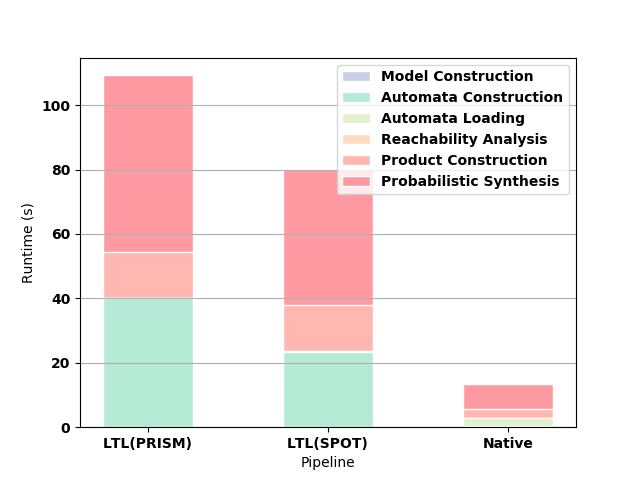}
        \vspace{-6mm}
        \caption{Hallway grid with \textbf{F9}.}
        \label{fig:times_hallway_F9}
    \end{subfigure}

    \begin{subfigure}[b]{0.32\textwidth}
        \centering
        \includegraphics[trim=15 5 30 30, clip, width=\textwidth]{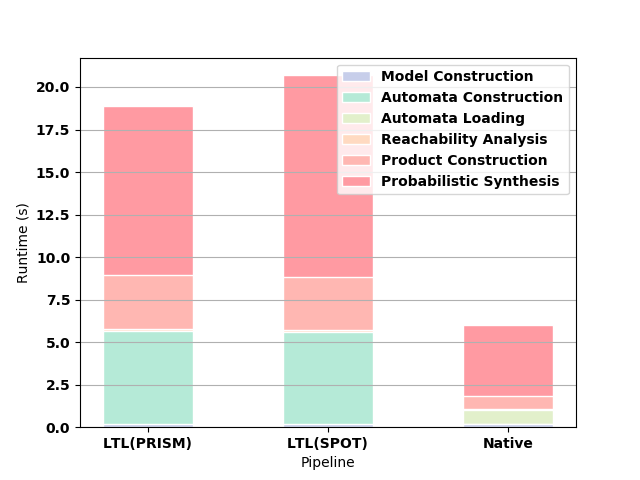}
        \vspace{-6mm}
        \caption{$10 \times 10$ grid with \textbf{OS}.}
        \label{fig:times_10x10_other}
    \end{subfigure}
    \begin{subfigure}[b]{0.32\textwidth}
        \centering
        \includegraphics[trim=15 5 30 30, clip, width=\textwidth]{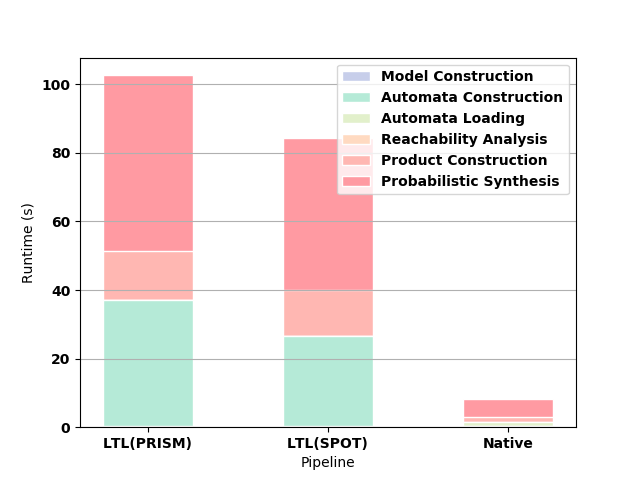}
        \vspace{-6mm}
        \caption{$10 \times 10$ grid with \textbf{OL}.}
        \label{fig:times_10x10_another}
    \end{subfigure}
    \begin{subfigure}[b]{0.32\textwidth}
        \centering
        \includegraphics[trim=15 5 30 30, clip, width=\textwidth]{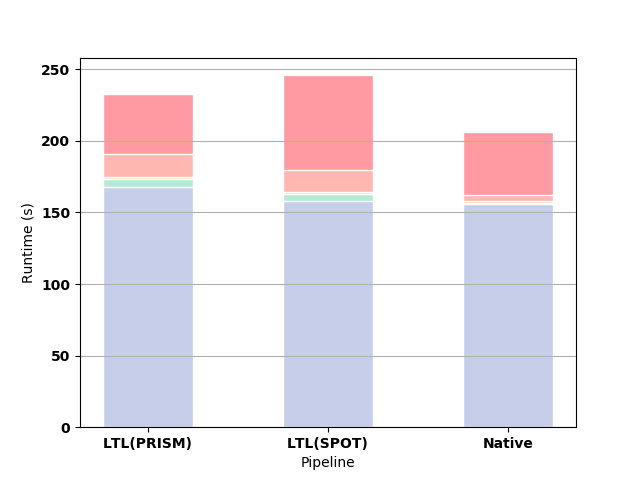}
        \vspace{-6mm}
        \caption{$50 \times 50$ grid with \textbf{F8}.}
        \label{fig:times_50x50_F8}
    \end{subfigure}
    
    \begin{subfigure}[b]{0.32\textwidth}
        \centering
        \includegraphics[trim=15 5 30 30, clip, width=\textwidth]{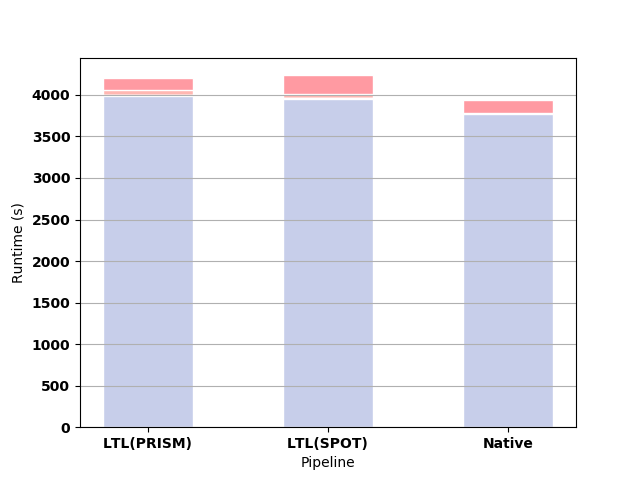}
        \vspace{-6mm}
        \caption{$100 \times 100$ grid with \textbf{F8}.}
        \label{fig:times_100x100_F8}
    \end{subfigure}
        \begin{subfigure}[b]{0.32\textwidth}
        \centering
        \includegraphics[trim=15 5 30 30, clip, width=\textwidth]{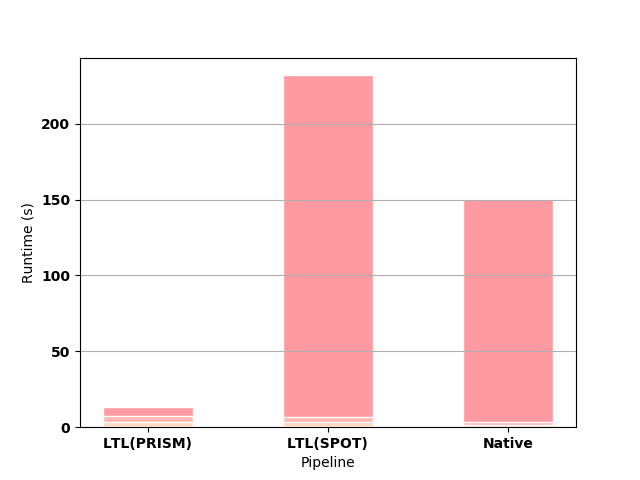}
        \vspace{-6mm}
        \caption{$5$ Dining Philosophers with \textbf{D5}.}
        \label{fig:phil5F5}
    \end{subfigure}
    \begin{subfigure}[b]{0.32\textwidth}
        \centering
        \includegraphics[trim=15 5 30 30, clip, width=\textwidth]{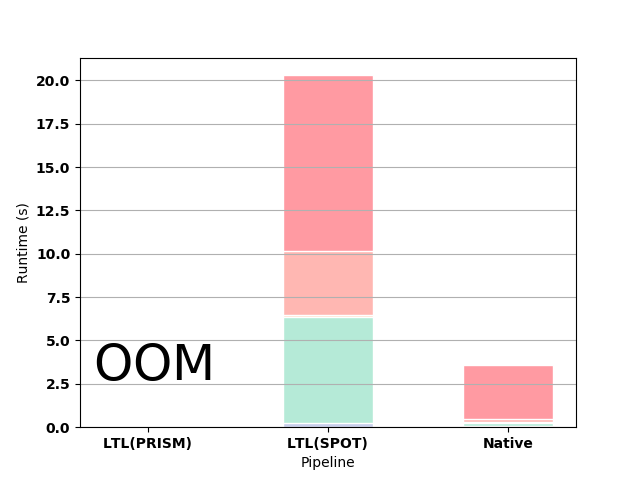}
        \vspace{-6mm}
        \caption{$5$ Dining Philosophers with \textbf{D5C}.}
        \label{fig:phil5complex}
    \end{subfigure}
    \caption{Runtime results for the Gridworld and Dining Philosophers MDPs with the LTL (PRISM and SPOT) and native pipelines.}
    \vspace{-4mm}
\end{figure*}

\subsection{Experimental Results}


We provide experiments varying the complexity of the formulas and of the MDPs.
Total runtime is shown in seconds. We refer to the translation to \ltl as the ``\ltl pipeline'' and the translation to a DFA via MONA as the ``Native pipeline.'' All plots are best viewed in color.


\subsubsection{Automata Construction}


First we measure how the length of the formula affects the synthesis computation time.  For all of these experiments, we use a $10 \times 10$ Gridworld as the original MDP. For small formulas, the time needed to read the HOA file from hard disk outweighs the shorter construction time and smaller automata of the native approach as shown in \autoref{fig:times_10x10_F3}. For formulas longer than \textbf{F3}, the native approach offers better computation time (e.g., see \autoref{fig:times_10x10_F8} and \autoref{fig:times_10x10_F9}).

PRISM successfully builds automata for formulas up to \textbf{F9}. In \autoref{fig:times_10x10_F10}, we highlight the superiority of the native approach in constructing the automaton for \textbf{F10}. For this formula, PRISM runs out of memory when constructing the automaton according to the \ltl pipeline. SPOT completes for \textbf{F10} but runs out of memory on \textbf{F11}. MONA works for formulas up to \textbf{F17}, which takes 8.586 seconds to compute, though writing the resulting file to disk is prohibitively expensive (the file is larger than 10GB). MONA runs out of memory constructing the automaton for \textbf{F18}.



We also considered a Gridworld with random obstacles (\autoref{fig:times_random_F9}) and with hallways (\autoref{fig:times_hallway_F9}), both with the formula \textbf{F9}. Finally, we consider two other formulas (\textbf{OS} and \textbf{OL}) in \autoref{fig:times_10x10_other} and \autoref{fig:times_10x10_another} to demonstrate that our improvement is not tied to the specific form of the specification.

For the Dining-Philosophers example, we consider 
five philosophers. Again, our results show that the native pipeline is faster and more memory efficient than the other approaches for automata construction (\autoref{fig:phil5complex}). However, the Dining-Philosophers example illustrates an issue where PRISM's built-in automata translation sometimes constructs automata such that computing the maximal accepting end-component using BDDs within PRISM is significantly faster (\autoref{fig:phil5F5}). The automata construction is still slower than both SPOT and our native approach, and the automata returned have more states, but the BDD representation is more efficient. This issue affects some automata, not only for our tool, but also for other external tools we tested. This means not only automata construction and size, but also the BDD representation are important for overall runtime. However, PRISM's built-in translator's memory usage scales more poorly than SPOT or our native approach. On a more complex formula \autoref{fig:phil5complex}, PRISM's built-in construction runs out of memory, even though SPOT and the native approach can finish construction in less than ten seconds.  Thus, for sufficiently complex formulas, PRISM's built-in translation is not viable. This suggests a need for automata construction methods that are not only more efficient but also produce automata that work well within PRISM.

\begin{wrapfigure}{r}{0.66\textwidth}
    \vspace{-4mm}
    \centering
    \begin{subfigure}[b]{0.31\textwidth}
        \centering
        \includegraphics[trim=15 5 30 30, clip, width=\textwidth]{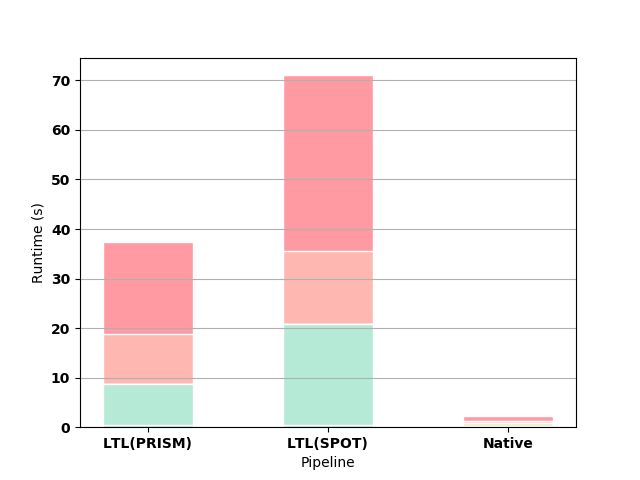}
        \vspace{-5mm}
        \caption{$200$-height random Nim game.}
        \label{fig:nim}
    \end{subfigure}
    \begin{subfigure}[b]{0.31\textwidth}
        \centering
        \includegraphics[trim=15 5 30 30, clip, width=\textwidth]{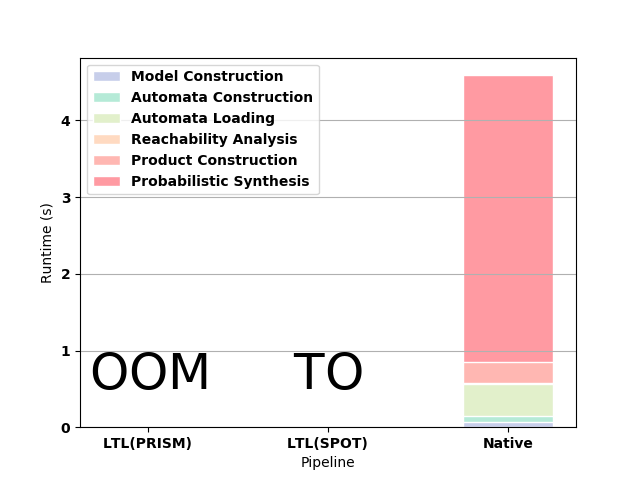}
        \vspace{-5mm}
        \caption{Double-Counter game.}
        \label{fig:counter}
    \end{subfigure}
    \vspace{-2mm}
    \caption{Runtime results for the Nim and Double-Counter games.}
    \vspace{-3mm}
\end{wrapfigure}
In the Nim game, the native approach far outperforms both PRISM and PRISM with SPOT (\autoref{fig:nim}). In the Double-Counter game, PRISM runs out of memory and PRISM with SPOT times out after more than 3 hours, whereas the native approach completes synthesis is less than 5 seconds.

Overall we observe that for large formulas the native pipeline offers significantly better scalability than the LTL pipeline (tested with PRISM's built-in translator, SPOT, Rabinizer, and LTL3DRA automata translators).  With an implementation that does not access the hard disk, we expect even better performance of the native pipeline.

\subsubsection{Automata Size}
The DFA generated by the native pipeline is minimal and typically much smaller than the \ltl pipeline's DRA. SPOT and PRISM typically produce similarly sized DRAs. Because we take the product of the input MDP with these automata, we expect the size of the resulting product to be smaller in the native pipeline.  \autoref{table:automata_size} shows the automaton sizes for the various formulas.  \autoref{table:dra_product_size} shows the sizes of the product MDPs for the \ltl and native pipelines respectively.

Interestingly, while the size of the DFA obtained from MONA is roughly half the size of the DRA constructed by PRISM, the final sizes of the products of the MDP and the automata are comparable for both approaches. The number of states, transitions and nondeterministic choices are all measured after reachability analysis is performed. We see time savings in the reachability analysis and product construction phases (both about two times faster), but the final products are similar in size. However, the product from the native pipeline has fewer nondeterministic transitions than the product from the \ltl pipeline. The improvements from this are small relative to the time it takes to construct a large \added{gridworld} MDP \added{(\autoref{fig:times_100x100_F8})}.


Examples of this are shown in \autoref{fig:times_50x50_F8} and \autoref{fig:times_100x100_F8}. Note that the majority of the computation time is spent constructing the MDPs.  For the \ltl pipeline, this construction takes slightly longer time because of the augmented MDP.
It is important to note that the native approach allows us to use larger formulas with these models whereas the \ltl pipeline is limited to \textbf{F10} or smaller.

In summary, we observe that the native pipeline is the most efficient in terms of both runtime and memory. There are two possible drawbacks to the native pipeline. First, it requires reading from disk, which for small formulas negates the advantage in automata construction time. Second, PRISM's built-in automata translation sometimes constructs automata whose BDD representations work better within PRISM than automata from any external tools we tested \added{(\autoref{fig:phil5F5})}. 

\begin{table*}
\centering
\caption{\small Sizes of the automata from \ltl (DRA) and native (DFA) pipelines.}
\label{table:automata_size}
\vspace{-2mm}
\footnotesize
\begin{tabular}{| l || c | c | r | c | c | c | c | c | c | c |}
\hline
             & \, \textbf{F3} \, & \, \textbf{F8} \, & \, \textbf{F9} \, & \, \textbf{F10} \, & \, \textbf{OS} \, & \, \textbf{OL} \, & \, \textbf{D5} \, & \, \textbf{D5C} \, & \, \textbf{Nim} \, & \, \textbf{Counter} \, \\ 
  \hline
  DRA states\, & 19  & 512  & 1,027 & NA & 515  & 1,028 & 67 & NA & 701 & NA\\  
  DFA states\, & 10  & 258  & 514   & 1,026 & 258  & 259 & 33 & 29 & 67 & 130\\
  \hline
\end{tabular}
\end{table*}

\begin{table*}
\vspace{-2mm}
    \centering
    \caption{Sizes of the product MDP in the \ltl pipeline and native pipelines.}
    \label{table:dra_product_size}
    \vspace{-2mm}
    \footnotesize
    \begin{tabular}{ |l || r | r | r | r | r | r | r |}
    \hline
                  & 10x10 \textbf{F3} & 10x10 \textbf{F8} & 10x10 \textbf{F9} & 50x50 \textbf{F8} & 100x100 \textbf{F8} & 10x10 \textbf{F10} & 10x10 \textbf{OS} \\
      \hline
      States (LTL)\;      & 890       & 24,678    & 48,998    & 641,478    & 2,568,978  & NA & 24,678    \\
      States (Native)\;     & 881      & 24,421    & 48,485    & 641,221   & 2,568,721  & 96357 & 24,421   \\  
      \hline

      Transitions (LTL)\; & 17,130    & 475,030   & 942,742   & 13,265,398 & 53,537,298  & NA & 475,030  \\
      Transitions (Native)\; & 16,256    & 450,468   & 893,760   & 12,623,936 & 50,968,336 & 1775424 & 450,368 \\
      \hline

      Choices (LTL)\;     & 4,410     & 122,358   & 242,934   & 3,206,358  & 12,843,858 & NA & 122,358 \\
      Choices (Native)\;     & 3,524     & 97,684    & 193,940   & 2,564,884  & 10,274,884 & 385428 & 97,684 \\

      \hline
      \hline
                    & \textbf{OL} & 10x10 rand & 10x10 halls & 5 Phil \textbf{D5} & 5 Phil \textbf{D5C}  & Nim & Counter \\
      \hline
      States (LTL)\;      & 24,935 & 44,253  & 19,476 & 4,548,220 & NA & 2115 & NA \\  
      States (Native)\;   & 24,519 & 43,740  & 19,187 & 1,476,976 & 93,068 & 404 & 449 \\  
      \hline

      Transitions (LTL)\;   & 475,287 & 804,097 & 310,996 & 46,948,710 & NA &  8343 & NA \\
      Transitions (Native)\;  & 452,172 & 759,856 & 291,536 & 7,891,750 & 494,420 & 1205 & 908\\
      \hline

      Choices (LTL)\;      & 122,615 & 219,209 & 96,220 & 44,059,330 & NA & 6297 & NA \\
      Choices (Native)\;    & 98,076 & 174,960 & 76,748 & 6,895,580 & 437,050 & 808 & 460\\

      \hline

    \end{tabular}%
    
    \vspace{-2mm}
\end{table*}

\section{Conclusion}
\label{sec:conclusion}
We introduced the problem of \ltlf synthesis for probabilistic systems and presented two approaches. The first one is a reduction of \ltlf to \ltl with a corresponding augmentation of the MDP. The second approach uses native tools to construct an automaton and then takes the product of this automaton with the MDP to construct a product MDP that can be used for synthesis through standard techniques. We showed that this native approach offers better scalability than the reduction to \ltl. Our work opens the door to the use of \ltlf synthesis on practical domains such as robotics, cf. \cite{he2019efficient}. Our tool is on GitHub \cite{wells2020github}.

For future work we would like to expand our results to include probability minimization (c.f. \autoref{theorem:ltlsynthesis}). Our native approach extends easily, but there are subtleties that prevent applying the \ltl pipeline to minimization. We are also interested in a fully symbolic methodology, where the automaton is represented symbolically and the product is also taken symbolically.  \replaced{The experiment in \autoref{fig:phil5F5} and our experiences with MoChiBa both indicate that implementing \ltlf methods optimized for PRISM may be more efficient than external tools.}{We would like to further investigate the issue we discovered where external tools quickly find automata with fewer states than PRISM's built-in method, but for which computing maximal accepting end-components in the product MDP takes much longer.} We are also interested in applying the work to robotics domains.


\paragraph{Acknowledgments}
Work on this project by AMW has been supported by NASA 80NSSC17K0162; by ML has been supported by the University of Colorado Boulder Autonomous Systems Interdisciplinary Research Theme and the NSF Centerfor Unmanned Aircraft Systems; by LEK has been supported in part by NSF grant IIS-1830549; and by MYV has been supported in part by NSF grants IIS-1527668, CCF-1704883, IIS-1830549, and an award from the Maryland Procurement Office. The authors would like to thank Gabriel Santos and Joachim Klein from the University of Oxford for their insightful comments and suggestions for efficient implementation of the framework in PRISM. We would like to thank Salomon Sikert from the Technical University of Munich for his aid in running MoChiBa.





\bibliography{refs}



\end{document}